\documentclass[12pt]{article}

\usepackage{arxiv}

\usepackage[utf8]{inputenc} 
\usepackage[T1]{fontenc}    
\usepackage{hyperref}       
\usepackage{url}            
\usepackage{booktabs}       
\usepackage{nicefrac}       
\usepackage{microtype}      
\usepackage{lipsum}         
\usepackage{graphicx}
\usepackage{subcaption}
\usepackage{natbib}
\usepackage{doi}
\usepackage{amsmath,amsthm,latexsym,amssymb,amsfonts,epsfig}
\usepackage{cleveref}       
\usepackage{verbatim}	
\usepackage{soul}	
\usepackage{xcolor}	
\usepackage{breakcites}

\newtheorem{theorem}{Theorem}[section]

\title{Fair Division Algorithms for Electricity Distribution}

\author{ 
Dinesh Kumar Baghel \\
Department of Computer Science\\
Ariel University, Israel\\
\texttt{dinkubag21@gmail.com} \\
\And
Vadim E. Levit \\
Department of Mathematics\\
Ariel University, Israel\\
\texttt{vadim.e.levit@gmail.com} \\
\And
Erel Segal-Halevi\\
Department of Computer Science\\
Ariel University, Israel\\
\texttt{erelsgl@gmail.com}\\
}

\begin{document}
\maketitle
\begin{abstract}
In many developing countries, the total electricity demand is larger than the limited generation capacity of power stations. Many countries adopt the common practice of routine load shedding --- disconnecting entire regions from the power supply --- to maintain a balance between demand and supply. Load shedding results in inflicting hardship and discomfort on households, which is even worse and hence unfair to those whose need for electricity is higher than that of others during load shedding hours. Recently, \cite{Oluwasuji2020} presented this problem and suggested several heuristic solutions. In this work, we study the electricity distribution problem as a problem of fair division, model it using the related literature on cake-cutting problems, and discuss some insights on which parts of the time intervals are allocated to each household. We consider four cases: identical demand, uniform utilities; identical demand, additive utilities; different demand, uniform utilities; different demand, additive utilities. We provide the solution for the first two cases and discuss the novel concept of $q$-times bin packing in relation to the remaining cases. We also show how the fourth case is related to the consensus $k$-division problem. 
One can study objectives and constraints using utilitarian and egalitarian social welfare metrics, as well as trying to keep the number of cuts as small as possible. 
A secondary objective can be to minimize the maximum utility-difference between agents. 
\end{abstract}


\section{Introduction}
Electricity is of utmost importance for the development of social and economic sectors. Electricity is perishable and should be consumed as soon as produced; it might be highly expensive to install alternative storage capacity. 
However, in this regard, developing countries face many challenges, like infrastructure, environmental, and sustainability issues \cite{Kaygusuz2012}. Without additional dedicated financial, institutional, and technological policies, some 15\% of the world's population will still expropriate electricity, a majority of which lies in developing Sub-Saharan African countries \cite{Kaygusuz2012}. This lack of electricity is a serious impediment to the development of industries and communities, thereby causing a significant drag on economic growth. In comparison, the United Kingdom generates on an average over 30GW of electricity\footnote{\url{https://gridwatch.co.uk}} for a population of about 68 million people\footnote{\url{https://worldpopulationreview.com/countries/united-kingdom-population}}, whereas Nigeria, a developing country, generates under 8.5GW\footnote{\url{https://www.ceicdata.com/en/indicator/nigeria/electricity-production}} of electricity for a population of over 211 million people\footnote{\url{https://worldpopulationreview.com/countries/nigeria-population}}. Electricity deficiency aggravates further with the growing population, development of backward areas, and the ever-increasing use of digital devices. Because of these factors, emerging countries will continue to face energy issues in the near future. Nigeria is one such example of an acute power crisis  \cite{GatugelUsman2015}. 

It is necessary to have a basic understanding of the electrical distribution grid abstractly. In a typical power distribution grid, power travels from the power plant to the transmission substation. It further travels to the distribution grid, where the power is stepped down so that it can be useful to homes and businesses. After necessary stepping down, the power leaves this grid in different directions via a distribution bus. Lines make up each bus on the network. Each of these lines depicts a group of individual consumers with varying electrical requirements \footnote{\url{http://www.science.smith.edu/~jcardell/Courses/EGR220/ElecPwr_HSW.html}l}.

Due to this significant difference in the electricity demand and generation capacity of emerging countries, and to maintain the desired frequency of electric current (which is 50Hz or 60Hz), disconnecting a substantial proportion of their electrical networks from the supply is the only choice for many developing countries. This event where a portion of an electricity network is deliberately disconnected from the supply is termed \emph{load shedding}. Load shedding, whose sole purpose is to maintain stable balance within the electricity grid, often results in disconnecting more load than is required and seriously impacts the comfort of the households in the shedded region. Load shedding will be present in the near future. Therefore, it is necessary to improve the electricity availability at the household level as much as possible when they need it the most on the condition that the cumulative electricity consumption by the households currently receiving electricity should not exceed the supply. In addition, this household-level electricity allocation reduces waste as compared to conventional load shedding approaches, improves fairness and revenue, and hence increases satisfaction within the system. \cite{Gaur2017} have described how a variety of factors can affect the electricity demand during different hours of the day. Therefore, any electricity allocation solution must consider these heterogeneous preferences in a complex power system.

In light of the above, we model households as agents and time as a resource, where each household has its preference for the consumption of electricity at each time.
The problem is similar to the classic problem of \emph{cake-cutting}, where agents have distinctly separate interests over different parts of the cake.  The objective is to come up with a fair allocation of the cake under these different and conflicting interests of the agents. Let us assume the time interval to be allocated among agents is $[0, T]$, the supply is $\mathcal{S}$, and agent $i$'s demand and utility at time $t$ are $d_i^t$ and $c_i^t$, respectively. Our objective is to compute an egalitarian electricity distribution while minimizing the number of cuts. 

There are four cases to consider, from simplest to most complex.
(1) Identical Demand, Uniform Utilities, (2) Identical Demand, Additive Utilities, (3) Different Demand, Uniform Utilities, and (4) Different Demand, Additive Utilities.
In the first two cases, where all agents have the same demand, the goal is to allocate the time-interval among $n$ agents as per their utility (or comfort or preference) vector such that, at any point in time, the cumulative demand of all agents connected to electricity should not exceed the supply $\mathcal{S}$. 
In the third case, we have shown it is not optimal to connect each agent $1/k$ of the time, where $k$ is the optimal number of bins formed by packing the $n$ agents' demands. This non-optimality leads us to devise the new notion of $q$-times bin packing. 
Finally, in the fourth case, we prove that egalitarian electricity division, in this case, may require $n-1$ cuts and is PPA-hard to compute.

\section{Related Work}
\cite{Pahwa2013} have proposed three load shedding strategies at the bus level: Homogeneous Load Shedding Strategy, Linear Optimization, and Tree Heuristic.
The simple, quick, and inefficient Homogeneous Load Shedding Strategy cuts a certain percentage of load from all buses in the system. It is useful for critical situations, but the amount of load shed is a lot higher than required. Another drawback of this strategy is that buses that may not be harmed by the original failure undergo load shedding.
Linear Optimization uses linear programming to meet as much demand as feasible in the system. However, this strategy results in very few buses in the system, shedding completely or a large percentage of the load. This is an unfair situation of providing electricity at the cost of a few buses. Moreover, this strategy is computationally expensive for large systems. 
They resolved this issue in the proposed tree heuristic strategy, which forms a tree based on the network elements. The tree is formed using the initially failed line. It then disconnects the same percentage from a subset of lines selected in the tree. However, if the formed tree is very small with the constituent lines of the tree not carrying enough load, then this strategy may not suffice to keep the system intact. Moreover, their result shows that this tree heuristic does not work for all the cases.
Although the proposed tree heuristic results in better load shedding as compared to the homogeneous strategy, it still disconnects the load at times when households are running indispensable activities.  		
Therefore, it is more important to take care of the electricity needs at the household level while developing solutions for electricity distribution.

In accordance with this,  \cite{Shi2015} developed a fair load-shedding solution where smart agents (buses on the system) communicate with their neighbors (upstream and downstream) and determine the amount of load they can shed and their corresponding compensation in real-time. Assuming the communicating agents are rational, they used linear models of incentives. Applying such a solution to the electric grid requires advanced information and communication technologies along with modern sensors and intelligent protection applications. However, developing countries require less complex solutions that deal with user-derived heterogeneous preferences. A retrofitted meter is a traditional meter with embedded units to provide the desired communication and control functionality. Recent advances in the design of smart retrofits aimed primarily at developing countries have made it possible to consider the needs of households within the network \hbox{\citep{Azasoo2015, Keelson2014}}. Smart retrofitted meters monitor household demands and plan for electricity ahead of time. Global System for Mobile Communications (GSM) technology, securely transmitting energy consumption data, remote connection, disconnection, and displaying user statistics are some of the features of smart retrofitted meters used to monitor the household demands and plan for electricity ahead. In the same line, \cite{Heggie2018} proposed a solution to allocate as much electricity as possible among the regions in the distribution network. The regulator sets the target proportion of supply to be delivered to each region. The allocation method will try to fulfill two objectives: first, minimizing the load shed, and second, fulfilling the target proportion set for each region as max as possible. At a more granular level, \cite{Azasoo2019} proposed a mechanism that bridges the gap in demand and supply by reducing the energy consumption of user-defined low priority tasks at the household level during peak hours using smart metering technology. They assumed that agents would report their true priorities of appliance usage. It is unfair since it may prohibit agents from performing some essential tasks (assuming all tasks are set to the same priority and hence, equally important) by reducing energy consumption.

 In light of this, \cite{Oluwasuji2018} presented four heuristic algorithms for fairly disconnecting households from the supply. They focused on connecting households as evenly as possible in terms of the number of hours. In evaluating the performance of these alternatives, they used utilitarian, egalitarian, and envy-freeness social welfare metrics. In continuation with this, \cite{Oluwasuji2020} suggested a more fair heuristic method. They formulated the load shedding problem as a multiple knapsack problem (MKP) and solved it using integer linear programming. 
 A significant problem with their approach is that they do not consider the number of cuts in the allocation since this constraint cannot be integrated into their ILP model.
 
As a result, the fair load shedding (or electricity distribution) problem can be represented as a fair resource allocation problem where electricity is the resource to be shared at the household level. As a resource allocation problem \cite{Chao2016} solved the electricity trading issue in a smart grid. From the commercial aspect of a smart grid, they proposed a fair electricity trading method among buyers (smart grid agent demand) and sellers (energy generated from renewable resources within the smart grid). For the buyer, electricity is cheaper to trade than to buy from a power utility. In the case of the seller, it is more profitable to sell other agents within the grid than to sell to a power utility. In addition, \cite{Gerding2011} solved the hybrid electric vehicles (EVs) charging problem based on the preferences reported by the agents (EV owners). Their developed model-free mechanism solved the problem of coordinating EV charging into the electricity grid to prevent overloading. They assumed that valuations are non-increasing for each incremental unit of electricity. Moreover, their mechanism results in certain units remaining unallocated to ensure truthfulness. However, these unallocated units are returned to the grid for other use, but in the case when the purpose is to distribute electricity to households, these unallocated units are necessarily wasted. \cite{Stein2012} extended this approach for pure EVs to model-based online mechanism design, which fulfills the EVs preferences using the pre-commitment notion, i.e., the mechanism is guaranteed to fulfill the demand for the selected agent. Their model assumes that agents have a value for a particular amount of resource, which, in contrast to our model, does not have an additive value for receiving more and values it to $0$ for receiving less. \cite{Miller2012} proposed a multi-agent-based coordination algorithm to optimally incorporate the power outputs of renewable generators in the distribution network. The aforementioned approaches require that agents behave rationally and genuinely state their values to arrive at a solution, whereas, in our model, we are generating these values centrally.

\cite{Chen2012} addressed the issue of fairly allocating cooling to households when the required power to achieve cooling surpasses supply. They predicted the relationship between power and cooling and proposed the min-max and proportional fairness approach. Min-max approach determines the common temperature to set for all the agents. As a result, some agents get the requested cooling, whereas others are at a greater inconvenience. On the other hand, the proportional approach allocates a fraction of their requested temperature and hence the required power to achieve it. However, it is prone to manipulation as agents can demand a lower temperature than they require. 

A different related problem where the amount of available resource is variable has been studied by \cite{Buermann2020}. In their setting amount of available energy is distributed at every time step, whereas our focus is on when to allocate energy. In their paper, the valuation function of the agents is a linear satiable function. It shows that an agent derives some usefulness even in the case of partial fulfillment of the demand, whereas in our case of piecewise constant functions, an agent's utility is significant only if the requested demand is fulfilled; otherwise, it is zero. Practically we can explain it as follows: suppose at some time a household is running some appliance that requires 1kW of electricity to operate; otherwise, it will not function. During that time, this was the only activity performed by the household. In this case, during that time, there is no point in fulfilling the partial demand of the agent as it is of no use to him. It is as good as giving no electricity.

In another paper, researchers have combined game-theoretic bankruptcy rules with Nash bargaining to solve the power allocation problem at the province level \cite{Janjua2021}. However, at the agent level, their solution is not feasible as it allocates a fraction of the demand to agents. It may be the case that in the absence of full demand, an agent may not be able to finish some critical work. \cite{Ali2021} have proposed a mechanism that solves the load shedding in the case when loads can be shifted to some other hour of the day. However, their solution is not practical in the real world as there are many activities that can not be shifted to a different time. In the mechanism design proposed in \cite{akasiadis2017mechanism}, the agents optionally commit to shift a fraction of their load by specifying their shifting capacity, shifting costs, and the non-peak interval to shift. The mechanism proposed in this paper shifts the peak load to some non-peak intervals. It does so by first defining a marginal shifting quantity corresponding to non-peak intervals and then selecting a group of agents whose sum is within this marginal quantity. Mechanism incentives agents for their truthful commitments. However, if necessary, this shifting can cause some critical work to shift to some other time resulting in a higher impact on their comfort hence making this approach less relevant. Moreover, it may happen that coalition strategies used by the mechanism can be unfair as it may select a few agents again and again, and some agents are never selected.

We focus on cake cutting as a classic resource allocation problem to solve the fair electricity distribution problem. The problem was first imposed by Steinhaus \cite{10.2307/1907319}. In cake cutting, a cake is a metaphor for the resource. However, we treat the time interval as a cake, unlike the previous approaches. A cake-cutting mechanism allocates the divisible resource to agents with different valuation functions (or preferences) according to some fairness criteria. A number of cake cutting protocols have been discussed in \citep{brams_taylor_1996, robertson1998cake, Rothe2016}.  In our problem, time-interval will be allocated among the agents without violating the supply constraint. The solution to our problem differs from the classic cake-cutting algorithms in the sense that the sum of the demands of all the agents connected at time $t$ should respect the supply constraint, and several agents may share the same piece. 

Our goal is to develop fair electricity distribution solutions for households as the residential sector forms a major part of the grid's demand. For example, in Nigeria, the residential sector represents 51.3\% of the grid demand \cite{Nwachukwu2014}. This residential electricity consumption is expected to rise further \cite{AmazuiloEzenugu2017} with the increase in living standards. Modeling the solution at the household level will result in better grid conditions and energy situations. We use the dataset representative of Nigeria's household energy consumption as in \cite{Oluwasuji2020}.

\section{Fair electricity distribution model for households and notations}

\subsection{Supply}
We denote by $\mathcal{S} \in \mathbb{R}_{> 0}$, the total supply of electricity available for distribution. It is (usually) measured in kilowatts.
We assume, for simplicity, that the supply is constant over time. 
In general, there can be alternative renewable energy sources on the network. The power generated from these resources is typically weather dependent and hence highly variable in nature, which makes it difficult to know the exact amount of power at the time of decision-making. 
Therefore, power generated from these resources is not considered in the current work.
In the future, we may consider a more general model, in which the supply can change with time.

We model the distribution of electricity as an allocation of a time-interval $\mathcal{C} = [0, T]$. For example, $T$ can denote a single hour, day, week, etc.

\subsection{Demand}
We model each household as an agent. We denote the set of agents by $[n]$, where $[k] = \{1, \ldots ,k\}$, and $n$ is the number of agents.
We define, for each agent $i$, the hourly demand at time $t$ as $w_i^t$. 
Initially, we assume that the demand is static, that is, $w_i^t = w_i$ for all $t\in[0,T]$.
We hope to generalize the model to time-dependent demand in the future.

Each agent $i$ should receive a subset $X_i$ of the cake $[0, T]$, which means that the agent is connected to electricity during the time represented by $X_i$.
We denote an allocation by $X := (X_1,\ldots,X_n)$.
In the classic cake-cutting problem, the pieces allocated to different agents must be pairwise-disjoint. 
This is an important way in which our problem generalizes the cake-cutting problem: in electricity division,  several agents may be connected at the same time. The requirement is that, at each point in time $t\in [0,T]$, the total demand of all agents connected at time $t$ is at most the supply $\mathcal{S}$.

We assume that there are no transmission constraints imposed by the grid to distribute the supply generated. That is, the only limit on the allocation is the total supply $\mathcal{S}$.

\subsection{Utility}
Each agent $i$ has a utility function $u_i$, which assigns to each time-interval $Z\subseteq [0,T]$, the utility that the agent gains from being connected to electricity at $Z$.%
\footnote{
\citet{Oluwasuji2020} use the term ``comfort'' instead of utility. They derive the utility of agents by averaging the demand over the past four weeks, and normalizing it by dividing by the maximum value. We consider utility and demand to be two different and independent inputs.
}

We assume that the agent utilities are additive over time, that is, $u_i(Z \cup Y) = u_i(Z) + u_i(Y)$ if $Z$ and $Y$ are disjoint time-intervals.
We also assume that $u_i$ is non-atomic, that is, $u_i([t,t]) = 0$ for all $t\in [0,T]$.

\subsection{Fairness and Efficiency Criteria}
We use three measures for the quality of an allocation:
\begin{itemize}
\item Its egalitarian welfare $eg$, defined as:
$eg :=  \underset{i \in [n]}{\min} \; u_i(x)$;
\item Its utilitarian social welfare $ut$, defined as: $ut := \sum_{i=1}^{n} u_i(x)$;
\item 
The maximum utility-difference $ef := \{\max_{i,j} \{\lvert u_i(x) - u_j(x)\rvert\}\}$. 
\end{itemize}

\subsection{Information considerations}
We assume that the demand of each agent is known to the divider. This is a reasonable assumption, as often the households are equipped with smart meters, from which we can compute their demands.
Additionally, the demand of each agent is often fixed in advnace acccording to the type of connection between the agent's house and the power grid.
Therefore, there are no strategic issues in reporting the demands.
In the future, one may consider mechanisms that incentivize agents to truthfully report the actual amount of electricity that they need, but in this work we assume that the demand is given.

We are also not considering strategic considerations in reporting the utility functions. 
Recently, \cite{Tao2021} has proved the nonexistence of a truthful cake cutting mechanism even in the simple case of two agents with piecewise-constant and strictly positive utilities. 
Truthful cake-cutting protocols are known only for very special cases \citep{Chen2013, Bei2020, Alijani2017}.

Since our problem is a generalization of cake-cutting, truthful mechanisms are definitely not possible.

\subsection{Key assumptions}
We make the same assumptions as in \mbox{\cite{Oluwasuji2020}} to solve our problem. For completeness, we describe these assumptions here:
\begin{enumerate}
	\item The households are equipped with smart retrofitted meters that allow to connect or disconnect individual households at each time.
	\item The estimates of demands we receive from households are correct.
	\item The utilities are unrelated to electricity distribution events: agents may not be aware of upcoming energy-distribution events, and so do not conduct some operations in advance.
\end{enumerate}

\section{Fair Electricity Distribution Algorithms}
In this section, we will address the fair electricity distribution problem from the simplest to the most general case. 
We discuss the following cases:
\begin{enumerate}
\item Identical Demands, Uniform Utilities
\item Identical Demands, Additive Utilities
\item Different Demands, Uniform Utilities
\item Different Demands, Additive Utilities
\end{enumerate}
\subsection{Identical Demands, Uniform Utilities}
In this simple version of the problem, all agents have the same demand $d$, and their utilities are uniform for each time interval. Hence, agents' only interest lies in how much time they are connected. Assuming the total demand is greater than the supply, only $q = \lfloor \mathcal{S}/d \rfloor$ agents can be connected at any point in time. 
Each agent is connected a fraction $q/n$ of the time, and therefore each agent's utility is $q/n$ of the total utility. 

\subsection{Identical Demands, Additive Utilities}
In this case, agents still have identical demands, but may assign different utilities to different time-intervals of the cake $\mathcal{C}$. Since the agents' demands are same, again only $q = \lfloor \mathcal{S}/d \rfloor$ agents can be connected at any point in time. So, we can create $q$ adjacent copies of the cake $\mathcal{C}$. Let us call this cake $\mathcal{C'}$. Now divide the cake $\mathcal{C'}$ using the Even-Paz cake-cutting algorithm \cite{Even1984}. 
The algorithm guarantees that the allocation of $\mathcal{C'}$ is \emph{proportional}, which means that the utility for an agent $i$ is $U_i \geq \frac{1}{n}U_i(\mathcal{C'} \geq \frac{q}{n}U_i(\mathcal{C})$. 
Thus, we can guarantee the same egalitarian welfare as in the case of identical demands and uniform utilities.

As an example assume there are $n = 4$ agents, each having a demand of $2$. Assume cake $\mathcal{C}$ is $[0,2]$ and the supply is $\mathcal{S} = 4$. 
The agents' utilites are piecewise-constant:

\begin{center}
	\begin{tabular}{|c|c|c|}
		\hline
		Piece & [0,1] & $[1,2]$ \\ \hline
		Agent 1: & 0.8 & 0.2 \\ \hline
		Agent 2: & 0.2 & 0.8 \\ \hline
		Agent 3: & 0.7 & 0.3 \\ \hline
		Agent 4: & 0.3 & 0.7 \\ \hline
	\end{tabular}
\end{center}
Only $q = \lfloor 4/2 \rfloor = 2$ agents can be connected at any point of time. Now, we construct a cake from $2$ adjacent copies of the interval $\mathcal{C}$. Let's call this cake as $\mathcal{C'} = [0,4]$. 

Correspondingly, agents utilities will consist of 2 adjacent copies of their respective original utilities, i.e. $u_1^{'} = (0.8,0.2,0.8,0.2), u_2^{'} = (0.2,0.8,0.2,0.8), u_3^{'} = (0.7,0.3,0.7,0.3), u_4^{'} = (0.3,0.7,0.3,0.7)$ Upon applying the Even-Paz cake cutting on $\mathcal{C'}$ and then mapping the allocation to original cake $\mathcal{C}$, agents allocations are as shown in Figure 1.
\begin{figure}
\centering
\begin{subfigure}[b]{0.55\textwidth}
\includegraphics[width=\linewidth]{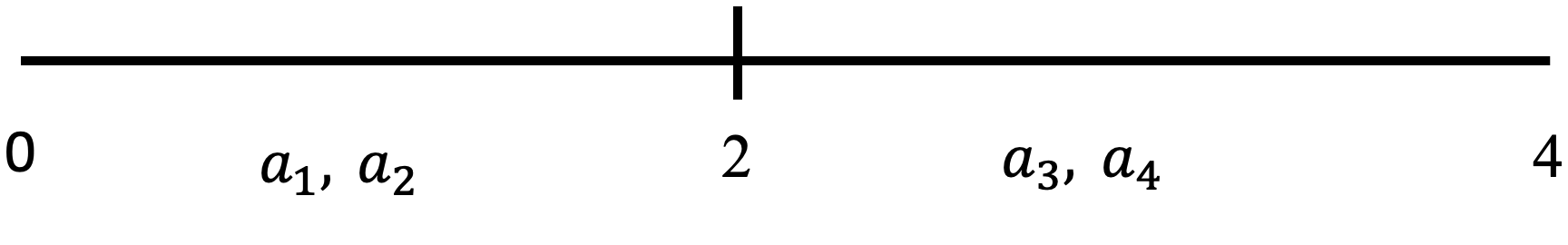}
\caption{} \label{fig:1a}
\end{subfigure}%

\begin{subfigure}[b]{0.55\textwidth}
\includegraphics[width=\linewidth]{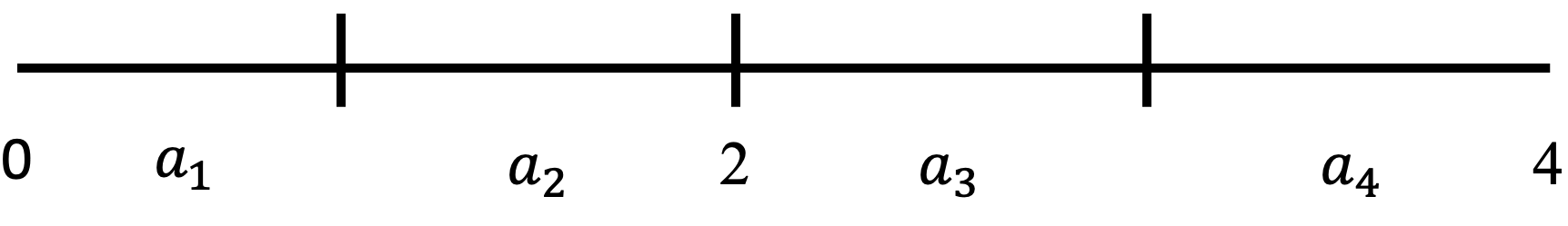}
\caption{} \label{fig:1b}
\end{subfigure}%

\begin{subfigure}[b]{0.35\textwidth}
\includegraphics[width=\linewidth]{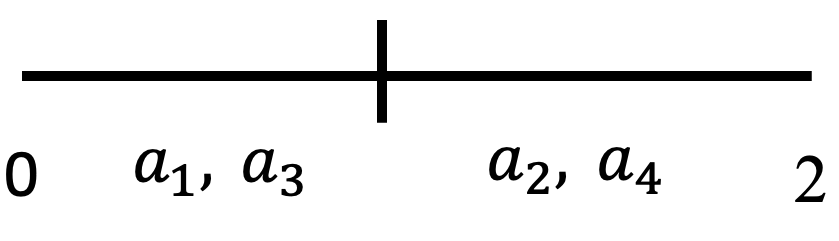}
\caption{} \label{fig:1c}
\end{subfigure}

\caption{Applying Even-Paz to $\mathcal{C'}$: (a) First Iteration, (b) Second Iteration. (c) Mapping $\mathcal{C'}$ to $\mathcal{C}$ and the corresponding allocations.} \label{fig:1}
\end{figure}
\subsection{Different Demands, Uniform Utilities}
A slightly more complicated version emerges if agents have different demands but uniform utilities. For example, we have $n = 3$ agents with uniform utilities. Let's assume supply is $\mathcal{S} = 30$ and the agent's demands are $10, 20$ and $30$, respectively. 

We define a \emph{maximal feasible set} as a set of agents, the sum of whose demands do not exceed the supply, and additional agents cannot be added to it without exceeding the supply.
The maximal feasible sets in the above example are: $\{10, 20\}, \{30\}$.
Obviously, only maximal feasible sets should be considered for allocation, otherwise there is unused supply.
 Now the question is \emph{what is a fair allocation of time in this case?} 

This problem can be represented as a problem of \emph{fractional approval voting}, where each of the agents $i \in [n]$ approves one or more candidates, and the output is a distribution over all the candidates. Fractional approval voting and the equivalent terms have been studied in \citep{Aziz2019, Bogomolnaia2005, Duddy2015, Brandl2021}.
In our case, the candidates are the maximal feasible sets, and each agent $i$ ``votes'' for all candidates containing $i$.

The most basic fairness condition in fractional approval voting is Individual Fair Share (IFS). IFS says that each agent $i$ should get a utility of $\geq 1/n$ of his/her total utility. In the cake-cutting literature, this condition is called ``proportionality''.
In our case, it means that each agent should be connected at least $1/3$ of the time. There are many allocations that meet this condition in our example, e.g. allocating $1/2$ of the time to each maximal feasible set, or allocating $2/3$ to the set $\{10,20\}$ and $1/3$ to the set $\{30\}$.

But, the IFS condition is too weak:
as seen in the case of equal demands, it is possible to give each agent a utility of $q/n$, which is usually much larger than $1/n$. There are two ways to strengthen IFS: Egalitarian and Group Fair Share. 

One way is to form groups of agents and treat each group as an individual agent. Our goal is to give each agent an allocation $\geq r$, where $r$ is the largest fraction possible given the agent demands.
		We call this approach the Egalitarian approach. 
In the example above, it allocations $1/2$ to each maximal feasible set.

Another way is the notion of Group Fair Share (Group FS or GFS) which says that for a group of $q$ agents, the total allocation should be at least $\geq q/n$
In the example above, it allocations $1/3$ to the $\{30\}$ and $2/3$ to the $\{10,20\}$.

To illustrate a potential problem with the GFS criterion, consider the scenario where the groups are $\{\epsilon,30-\epsilon\}$ and $\{30\}$. GFS still results in allocating the groups $2/3$ and $1/3$ respectively, despite the fact that $30-\epsilon$ and $30$ are almost the same, and hence this allocation is apparently unfair for the agent with a demand of $30$. Moreover, for a group like $\{5,10,15\}$ and $\{30\}$, the GFS allocation will turn out to be $3/4$ and $1/4$ respectively, which is even more unfair for a large family having demand $30$.

Minimalistically, we are interested in the ``equal treatment of equals" fairness principle. Therefore, in the egalitarian approach, our goal is to maximize $r$ so that each agent can be connected at least $r$ of the time. Finding $r$ is an NP-hard problem (by reduction from Partition). In contrast to fractional approval voting, the number of candidates is exponential in $n$. Therefore, it is not possible to check all candidates. It leads us to look for other techniques like Bin Packing algorithms. In bin packing, we are given a set of items (agents), and the goal is to pack this set of items into a minimum number of bins of size $\mathcal{S}$, where the size of each agent is its demand.
Let us assume the number of bins in a bin packing solution for a problem instance $I$ is $k$. Then each agent in $I$ can be connected $1/k$ of the time. However, this lower bound is not optimal. For example, if $\mathcal{S} = 2$ and the agent demands are $ = {1,1,1}$, then bin packing results in $k = 2$, which gives a lower bound of $r = 1/2$. But $r = 2/3$ is possible, by connecting $\{1,2\}$, $\{2,3\}$ and $\{3,1\}$ for $1/3$ of the time each.

This leads us to a new problem, which we call \emph{$q$-times bin packing}: given the set of items, the goal is to pack these items into $k$ bins, of size $\mathcal{S}$ each, such that each item appears in $q$ different bins. We can connect each bin $1/k$ of the time; thereby, each agent is connected for at least $q/k$ of the time. Therefore, the following research questions arise:
\begin{enumerate}
\item Which bin-packing algorithms can be adapted to $q$-times bin-packing?
\item Is the $q$-times bin-packing lower bound always tight for some $q \geq 1$?
\item Which efficient algorithms can be used to find a group-FS allocation?
\end{enumerate}

\subsection{Different Demands, Additive Utilities}
In the more realistic scenario, agents have different demands with additive utilities. 
What fairness guarantees are reasonable in this setting?

Let $r$ be the egalitarian value of the simplified problem in which agents have \emph{uniform} utilities. Can we guarantee each agent a utility $\geq r$ with \emph{additive} utilities? 
The answer is positive: the utilities in each minute are essentially uniform, so we can partition each minute using the egalitarian solution, and guarantee each agent a fraction of at least $r$ in each minute.

Nevertheless, the problem is that this allocation requires too many cuts. 
However, due to practical considerations, we would like to find a solution involving fewer cuts. 
The previously posed question can now be rephrased: Can we guarantee each agent a utility $\geq r$ with \emph{a small number of cuts}?

Let us assume that we have an egalitarian solution $r = q/k$ where $k$ is the number of bins from $q$-times bin-packing. Insights from the $k$-consensus division problem can help answer the question. A \emph{$k$-consensus division} is a partition of cake into $k$ pieces such that each agent $1, \ldots ,n$ values those pieces at exactly $1/k$. 
This partition can be achieved using a bounded number of cuts ($\#$cuts $\leq n)(k-1)$)\cite{Alon1987}.

Now, we can find a $k$-consensus division among agents $1,\ldots,n-1$ using $(n-1)(k-1)$ cuts, and let agent $n$ choose the $q$ pieces he prefers. 
We can now connect each bin (--- maximal feasible set) $1/k$ of the time, and since each agent appears in $q$ bins, the utility for agent $n$ is $\geq q/k$, and all other agents will get exactly $q/k$. 
However, the problem is NP-hard when $k = 2$, even when all valuations are piecewise constant. The following results hold for $k=2$:
\begin{itemize}
\item Computing an approximate $k$-consensus division with the minimal number of cuts ($n(k-1)$) is PPA-hard \citep{Filos-Ratsikas2018a, Filos-Ratsikas2020}.
\item Deciding whether there exists a $k$-consensus division with $n(k-1) - 1$ cuts is NP-hard \citep{Filos-Ratsikas2018b}
\end{itemize}
For $k=3$, we know that consensus division is PPAD-hard \citep{Filos-Ratsikas2020}. However, no such results are known for $k>3$ since more cuts are available. We only know that the problem is in PPA-k for any $k$ which is a power of a prime integer \citep{Filos-Ratsikas2021}. These hardness results can be summarized as:
\begin{center}
\begin{tabular}{|c|c|c|c|}
\hline
\textbf{\#pieces $k$} & \textbf{Agents} & \textbf{Cuts} & \textbf{Hardness} \\	\hline
2 & $n$& $n$ & PPA-hard \\	\hline
3 & $n$ & $2n$ & PPAD-hard \\	\hline
prime power & $n$ & $n(k-1)$ & belongs to PPA-k \\
\hline
\end{tabular}
\end{center}
So the questions are:
\begin{enumerate}
\item[(a)] Can electricity division be done with \textit{fewer} cuts?
\item[(b)] Can electricity division be done more \textit{efficiently}?	 
\end{enumerate}

Formally, we define the \emph{egalitarian electricity division problem} as follows.

A solution to the electricity division problem is a distribution, $E$, of the time intervals for each agent. The \emph{egalitarian welfare} of this distribution is the agent with the lowest utility: $eg(E) = \min_{i \in [n]}  \{u_i(E)\}$. 

INPUT: $n$ agents with different demands and different valuations over a time-interval; supply $\mathcal{S}$.

OUTPUT: a feasible allocation of connection time among the agents, such that $eg(E)$ is maximum among all feasible allocations.

\begin{theorem}
\label{4.1}
Egalitarian electricity division among $n$ agents:
\begin{enumerate}
\item[(a)] may require $n-1$ cuts
\item[(b)] is PPA-hard to compute.
\end{enumerate}
\end{theorem}
\begin{proof}
The proof is by reduction from $2$-consensus division.

Given a $2$-consensus division problem with $n-1$ agents with valuations $v_1, \ldots, v_{n-1}$, construct an electricity division problem with $n$ agents:
\begin{itemize}
\item The supply $\mathcal{S} = n-1$;
\item There are $n-1$ agents with demands $1$ and valuations = $v_1, \ldots , v_{n-1}$;
\item There is one agent with demand $n-1$ and valuation $(v_1 + v_2 + \ldots + v_{n-1})/(n-1)$
\end{itemize}
In this instance, there are only two maximal feasible sets: $\{1,\ldots,n-1\}$ and $\{n\}$. 
Therefore, the egalitarian division can be written as $(X_1,X_2)$, where $X_1$ is the interval in which the set 
$\{1,\ldots,n-1\}$ is connected, 
and $X_2$ is the interval in which the set 
$\{n\}$ connected.

In each maximal feasible set, the valuation of all agents is the same. Therefore, by classic cake-cutting results, it is possible to guarantee each agent a utility of at least $1/2$.
So the optimal egalitarian value is at least $1/2$.

All agents in the first feasible set value the first piece as at least $1/2$, i.e.
\[\forall i = 1,\ldots,n-1: v_i(X_1) \geq 1/2\]
and by additivity, values the second piece as at most $1/2$:
\[\forall i = 1, \ldots, n-1: v_i(X_2) \leq 1/2\]
The second piece $X_2$ is given to the $n^{th}$ agent whose valuation is the average of the valuation of all other agents. Its value is at least the optimal egalitarian value: $v_n(X_2) \geq 1/2$, but $v_n(X_2)$ is the average of all $v_i(X_2): \forall i = 1,\ldots,n-1$. This implies $v_i(X_2) = 1/2, \forall i = 1,\ldots,n-1$. Again by additivity, 
\[\forall i = 1,\ldots,n-1: v_i(X_1) = 1/2\]
So $(X_1,X_2)$ solves the problem of $2$-consensus division. Thus, consensus division among $n-1$ agents is a special case of egalitarian electricity division among $n$ agents (when the egalitarian value is $1/2$).

It is known that a 2-consensus division among $n-1$ agents may require $n-1$ cuts \cite{Alon1987},
and it is PPA-hard to find an approximate egalitarian division. 
\end{proof}
\noindent This leads us to the following research questions:
\begin{itemize}
\item[(a)] Are $n-1$ cuts (which might be necessary by the reduction from consensus division) always sufficient for egalitarian electricity division?
\item[(b)] How many cuts are needed for a group-FS division?
\item[(c)] Can we heuristically find an egalitarian/group-FS division with few cuts on realistic instances?
\end{itemize}

\section{Conclusion and Future Work}
In this paper, we studied the fair electricity division problem. It is related to several classic problems, like proportional cake cutting, bin packing, fractional approval voting, and consensus division. We have considered the four different cases for demand and utilities, from the simple one to the more generic one. We have provided the solution for the rather simplistic case of identical demand with uniform/additive utilities. We have seen how the novel concept of $k$-times bin packing relates to the remaining cases of varying demand with uniform/additive utilities. We have also shown that in the more general case of different demand and additive utilities, egalitarian electricity distribution may require $n-1$ cuts and is PPA-hard to compute \ref{4.1}. Ideas that we will develop here may apply to other settings of social choice or fair division in which agents may have \textit{different sizes}.

From a broader point of view, this paper shows how fair electricity distribution and related problems are complex and should be studied more. Some of the results from our work can be applied to the settings where many agents organized in overlapping groups want to attend or access the resource of limited capacity \cite{arbiv2022fair}. Our basic model of fair electricity distribution can be extended to include more supply constraints, such as supply $\mathcal S$ changing over time; a supply network with a capacity constraint on link capacities; OR it can be extended to include more complex demands, such as demand $d_i$ changing over time, with price acting as an incentive to agents to consume less electricity during more overloaded hours. Another interesting direction would be extending this work to a more general valuation function class.

\section{Acknowledgment}
We want to thank Dr. Aris Filos-Ratsikas, Dr. Enrico Gerding, and Prof. Amos Yehuda Azaria for invaluable feedback and comments that helped improve this research proposal.

\bibliographystyle{unsrtnat}
\bibliography{researchprop2.bib}

\end{document}